\documentclass[9pt]{amsart}
\usepackage{amsfonts}

\usepackage{graphicx}
\usepackage{amscd}
\usepackage{amstext}
\usepackage{amsmath}

\usepackage{amssymb,latexsym}
\usepackage{epsfig}

\newtheorem{theorem}{Theorem}
\theoremstyle{plain}

\newtheorem{proposition}{Proposition}

\numberwithin{equation}{section}

\begin{document}
\title[Periodic Orbits]
{Periodic orbits of  mechanical systems with homogeneous polynomial
terms of degree five }

\author{Alberto Castro Ortega} \address[A. Castro Ortega]{ Depto. de Matem\'aticas, Fac. de Ciencias, UNAM, C. Universitaria, M\'exico, D.F. 04510}
\email{acospacy@yahoo.com.mx}
%\thanks{}
\date{}
 \keywords{}
%\thanks{}
\begin{abstract}
In this work  the existence of periodic solutions is studied for the
Hamiltonian functions
\begin{equation*}
    H=\frac{1}{2}\left(
    p_X^2+p_Y^2+X^2+Y^2\right)+\frac{a}{5}X^5+bX^3Y^2,
\end{equation*}
 where the first term consist of a harmonic oscillator  and  the
 second term
 are homogeneous polynomials of degree 5 defined by two real parameters $a$ and
$b$. Using the averaging method of second order  we provide the
sufficient  conditions on the parameters to guarantee the existence
of periodic solutions for positive energy and we study the stability
of these periodic solutions.
\end{abstract}
\section*{Published in Astrophysics and Space Science DOI: 10.1007/s10509-015-2612-0} \maketitle
\section{Introduction}

The Galactic dynamics is an area of the Astrophysics where recently
the application of results coming from other areas as the Celestial
mechanics and the Nonlinear dynamics  has gradually establish common
methods and results well documented (Boccaletti and Puccaco 1999)
and (Contoupoulus 2002). The global dynamics of galaxies is not a
simple question and represent an actual challenge for the
researches.
\bigskip

In the research of stellar systems, the perturbative methods of the
analytical mechanics provide a simple and comprehensive description
of the dynamics of systems which cannot be exactly solved with
accurate quantitative predictions, even at the simplest level of the
procedure. The approach used by the perturbative methods consist to
study the dynamics of the original physical system  by an
approximating integrable system. Usually, the approximating system
is an expansion in power series of the potential function  in terms
of the coordinates variables. There are many methods to construct
the approximating integrable system as the method of the Lie
transform which provides the normal form of the system which is the
most simplest form that the system can be (Belmonte et al. 2008).
\bigskip

Many interesting problems in galactic dynamics are modeled by
introducing Hamiltonian systems with two degrees of freedom $X$ and
$Y$ of the form
\begin{equation}\label{Galactic}
    H=\frac{1}{2}\left(\omega_1p_X^2+\omega_2p_Y^2\right)+V(X,Y),
\end{equation}
where $\omega_1$ and $\omega_2$ are the unperturbed frequencies of
oscillation along the $X$ and $Y$ axis respectively and $p_X$ and
$p_Y$ are the momenta conjugate to $X$ and $Y$. The particular
interest is to determine the properties of the orbital structure of
the systems with potential functions with reflection symmetry with
respect to the both axis, as examples, the potentials given by
\begin{equation}\label{potgalac}
V(X,Y)=\log\left(R^2+X^2+\frac{Y^2}{q}\right),\hspace{5mm}V(X,Y)=\sqrt{R^2+X^2+\frac{Y^2}{q}}-R^2,
\end{equation}
 where $R$ is the core radius and the parameter $q$  determines  the
ellipticity of the potential (Belmonte et al. 2008). In order to
study the dynamics associated  to the potentials (\ref{potgalac}) it
have been considered models where the potential function $V(X,Y)$ is
expanded as a truncated series in the coordinates $X$ and $Y$
\begin{equation}\label{truncated}
    H=\frac{1}{2}\left(\omega_1p_X^2+\omega_2p_Y^2\right)+\sum_{k=0}^N\sum_{j=0}^k
    C_{(j,k-j)}X^jY^{k-j},
\end{equation}
where the truncation degree is $N$ and the parameters $C_{(j,k-j)}$
are determined by the problem under study. Although the polynomial
models are simplified to be considered realistic, they provide
information about  periodic solutions or chaos. The use of maps with
polynomial models to describe the galactic motion is a useful tool
because the numerical integration is faster and allows to visualize
the corresponding phase space (Caranicolas and Vozikis 1999).
\bigskip

The polynomial Hamiltonian systems (\ref{truncated}) are an actual
research topic because its importance in the research of nonlinear
phenomena. In models of the dynamics of galaxies the potentials with
homogeneous polynomial terms of some degree $N$ have been studied by
(Caranicolas and Vozikis 2004) for the case of potential functions
with homogeneous polynomials of third degree and  (Contopulos 2002)
for potentials with terms of fourth degree. From the mathematical
point of view, the reader can consult the following references about
mechanical systems with polynomial potentials (Falconi and Lacomba
1996, 2009), (Falconi et al. 2007) and (Dorizzi and Grammaticos
1983).
\bigskip

There are an important  cases where the frequencies are equal to
one, as known  $1:1$ resonant cases. In the study of polynomial
Hamiltonian systems there is an interesting problem known as the
H\'enon-Heiles model (Heiles and H\'enon 1964), which is a simple
non-integrable Hamiltonian chaotic system. The interest in the
H\'enon-Heiles system was initially motivated by the study of the
existence of a third isolating integral of motion in certain
galactic potentials admitting an axis of symmetry. The Hamiltonian
proposed by H\'enon-Heiles to study the existence of the third
integral  is given by
\begin{equation}\label{HH}
    H=\frac{1}{2}\left(
    p_X^2+p_Y^2\right)+\frac{1}{2}\left(X^2+Y^2\right)+\frac{1}{3}X^3-XY^2,
\end{equation}
 which is an extension of the harmonic oscillator to the anharmonic case
   where the perturbation term is a
homogeneous polynomial of degree three. Although the model
associated to Hamiltonian (\ref{HH}) is simple, the potential
produce all the complexities obtainable in any chaotic system. A
wide class of three particle systems can be reduced to a
H\'enon-Heiles type Hamiltonian by considering only the first three
terms in the Taylor expansion.
\bigskip

Following the ideas of H\'enon-Heiles it can be considered other
models where the perturbation is a homogeneous polynomial of
arbitrary degree with a finite number of parameters. In this paper,
using the averaging method of second order, we prove for positive
energy sufficiently small the existence of periodic solutions  for
the Hamiltonian systems consist of a harmonic oscillator plus a
homogeneous potential of fifth degree with two terms and two real
parameters $a$ and $b$ given by
\begin{equation}\label{Hamiltonian}
    H=\frac{1}{2}\left(
    P_X^2+P_Y^2\right)+\frac{1}{2}\left(X^2+Y^2\right)+\frac{a}{5}X^5+bX^3Y^2.
\end{equation}
We observe in the Hamiltonian systems (\ref{HH}) and
(\ref{Hamiltonian}) that the presence of the terms $X^5$ and $X^3$
in the expansion accounts for the breaking of the reflection
symmetry with respect the $X$ axis, the models still have the
reflection symmetry with respect to the $Y$ axis. The class of
potentials studied in this paper have not chose with the aim of
modeling some particular galaxies, the objective is to study systems
which are generic in their basic properties.
\bigskip

Our main result on the periodic orbits of (\ref{Hamiltonian}) is
summarized in Theorem \ref{teoprin}. The periodic orbits are the
most simple non-trivial solutions of an ordinary differential system
and depending on the type of stability
 determine the dynamics in their neighborhood. The averaging
method gives a quantitative relation between the solutions of some
non autonomous differential system and the solutions of the averaged
differential system, which is an autonomous one, this method leads
to the existence of periodic solutions for periodic systems. The
averaging is with respect to the independent variable and the right
hand sides of these systems are sufficiently small, depending on a
small parameter $\epsilon$.  The problem of finding periodic
solutions of the perturbed differential system is reduced to find
zeros of some convenient finite dimensional function. We provide the
conditions under which the averaging theory guarantees the
persistence of periodic orbits under the perturbation of the
harmonic oscillator, and we find them as a function of the energy
and the parameters $a$ and $b$.
\bigskip

The averaging method at first and second order in the context of
(Buica and Llibre 2004) has used successfully to prove the existence
of periodic solutions for the generalized H\'enon-Heiles Hamiltonian
system (Carrasco and Vidal 2013) and for the generalized classical
Yang-Mills Hamiltonian system with two parameters(Jim\'enez-Lara and
Llibre 2011), which consist of a perturbation of the harmonic
oscillator by quartic homogeneous potentials. It is important to
remark that when the potential function is the form $V(X^2,Y^2)$ the
averaging method at first order provide the information about the
existence of periodic solutions, as examples, see (Jim\'enez-Lara
and Llibre 2011) and (Llibre and Makhlouf 2013). When the terms of
the expansion of the potential are homogeneous polynomials of degree
odd, the average system at first order vanishes in the period, hence
we proceed to use the second order averaging method, as an example,
see (Carrasco and Vidal 2013).

\bigskip
In order to get our main results this work is organized as follows.
In section \ref{eqmotion} the equations of motion are presented, a
small  parameter $\epsilon$ is introduced by a convenient rescaling
transformation. In section \ref{averaging}, the results of (Buica
and Llibre 2004) are presented and we apply a change of coordinates
in order to use the average method of second order. In section
\ref{secperiodic}, we study the conditions on the parameters $a$ and
$b$ such that the periodic orbits exist for positive energy and we
determine their stability.

\section{Equations of motion}\label{eqmotion}
The equations of motion associated to the system (\ref{Hamiltonian})
are
\begin{eqnarray}\label{HSys}
% \nonumber to remove numbering (before each equation)
  \dot{X} &=& P_X, \\\nonumber
  \dot{Y} &=& P_Y, \\\nonumber
  \dot{P}_X &=& -X-aX^4-3bX^2Y^2, \\\nonumber
  \dot{P}_y &=&-Y-2bX^3Y.
\end{eqnarray}
The Hamiltonian system (\ref{HSys}) is not into the normal form for
applying the averaging theory, in order to have a small parameter in
the Hamiltonian system (\ref{HSys}) we introduce the parameter
$\epsilon$   by the change of variables $(X,Y, P_X, P_Y)$ to $(x,y,
p_x , p_y )$ where $X = \epsilon^{\frac{1}{3}} x$, $y =
\epsilon^{\frac{1}{3}} Y$, $P_X = \epsilon^{\frac{1}{3}} p_x$ and
$P_Y = \epsilon^{\frac{1}{3}} p_y$ which is
$\epsilon^{-\frac{2}{3}}$-symplectic. So the system (\ref{HSys})
becomes
\begin{eqnarray}\label{HSys1}
% \nonumber to remove numbering (before each equation)
  \dot{x} &=& p_x, \\\nonumber
  \dot{y} &=& p_y, \\\nonumber
  \dot{p}_x &=& -x-\epsilon\left(ax^4+3bx^2y^2\right), \\\nonumber
  \dot{p}_y &=&-y-2\epsilon bx^3y.
\end{eqnarray}
The Hamiltonian function associated to the previous system is
\begin{equation}\label{Hameps}
    H=\frac{1}{2}\left(
    p_x^2+p_y^2\right)+\frac{1}{2}\left(x^2+y^2\right)+\epsilon\left(\frac{a}{5}x^5+bx^3y^2\right),
\end{equation}
we will consider small positive values of the energy in our model.
By the standard theory of Hamiltonian dynamical systems, for all
$\epsilon\neq0$ the original and the transformed systems
(\ref{HSys}) and (\ref{HSys1}) have essentially the same phase
portrait. Furthermore, for $\epsilon\approx0$ the system
(\ref{HSys1}) is close to an integrable one.

\section{Second order averaging method}\label{averaging}
In this section we introduce the necessary results from averaging
theory of second order for proving the statements of this paper. We
follow (Buica and Llibre 2004), particulary, we focus in Theorem 3.1
that we enunciated as Theorem 1.
\begin{theorem}\label{secondave}
We consider the following differential system
\begin{equation}
\dot{x}=\epsilon F_1(t,x)+\epsilon^2F_2(t,x)+\epsilon^3
R(t,x,\epsilon),
\end{equation}
where $F_1, F_2:\mathbb R\times D\to\mathbb R^n$, $R :\mathbb
R\times D\times (-\epsilon,\epsilon)\to\mathbb R^n$ � are
continuous functions, $T-$periodic in the first variable, and $D$ is
an open subset of $\mathbb R^n$. We assume that
\begin{enumerate}
\item $F_1(t,\cdot)\in C^1(D)$ for all $t\in\mathbb R$, $F_1$,$F_2$,$R$ and $D_xF_1$ are locally Lipschitz with respect
to $x$, and $R$ is differentiable with respect to $\epsilon$. We
define $f_1,f_2 :D\to\mathbb R^n$ as
\begin{eqnarray}\label{second}
f_1(z) &=& \int_0^T F_1(s, z) ds,\\\nonumber f_2(z)
&=&\int_0^T\left(D_zF_1(s, z)\cdot\int_0^sF_1(t , z) dt + F_2(s,
z)\right)ds,
\end{eqnarray}
and assume moreover that
\item for $V\subset D$ an open and bounded set and for each $\epsilon\in(-\epsilon,\epsilon) \ {0}$, there exists
$a_\epsilon\in V$ such that $f_1(a_\epsilon)+ \epsilon
f_2(a_\epsilon) = 0$ and $d_B(f_1 + \epsilon f_2,V, 0) = 0$.
\end{enumerate}
Then, for $|\epsilon| > 0$ sufficiently small, there exists a $T
-$periodic solution $\varphi(\cdot, \epsilon)$ of system
(\ref{second}).
\end{theorem}

To apply Theorem \ref{secondave}, we introduce the following changes
of variables in order to obtain a $2\pi-$periodic system. Now, let
$\mathbb{R}^+ = [0,�\infty)$ and $\mathbb{S}^1$ the circle. We do
the change of variables $(x, y, p_x , p_y )\to (r, \theta, \rho,
\alpha)\in \mathbb{R}^+\times \mathbb{S}^1 \times� \mathbb{R}^+
\times \mathbb{S}^1$ defined by
\begin{equation}\label{coordchange}
x=r\cos\theta,
\hspace{3mm}p_x=r\sin\theta,\hspace{3mm}y=\rho\cos(\theta+\alpha),\hspace{3mm}p_y=\rho\sin(\theta+\alpha).
\end{equation}
This change of variables is well defined when $r > 0$ and $\rho > 0$
and it is not canonical, so we lost the Hamiltonian structure of the
differential equations.
\bigskip

The fixed value of the energy (\ref{Hameps}) in polar coordinates is
\begin{equation}\label{Hamrpo}
 h=\frac{1}{2}\left(r^2+\rho^2\right)+\epsilon\Bigg(\frac{a}{5}r^5\cos^5\theta+br^3\rho^2\cos^3\theta\cos^2(\theta+\alpha)
 \Bigg),
\end{equation}
and the equations of motion (\ref{HSys1}) assume the form
\begin{eqnarray}\label{ecrrho}
 \dot{r}&=& -\epsilon\sin\theta\Big[ar^4 \cos^4\theta+3 b r^2 \rho^2 \cos^2\theta \cos^2(\theta+\alpha)\Big],\\\nonumber
\dot{\theta}&=&-1-\epsilon r\cos^3\theta\Big[a r^2\cos^2\theta+3 b
\rho^2  \cos^2(\theta+\alpha)\Big],\\\nonumber \dot{\rho}&=&
-\epsilon b r^3\rho \cos^3\theta \sin2(\alpha + \theta),\\\nonumber
\dot{\alpha}&=&\epsilon r \cos^3\theta \Big[ ar^2
\cos^2\theta+b(3\rho^2-2r^2)\cos^2(\alpha+\theta)\Big],
\end{eqnarray}
where the derivatives of the left hand side of the above equations
are with respect to the time variable $t$, which is not periodic.
\bigskip

In order to write the system (\ref{ecrrho}) as a $2\pi$-periodic
differential system we consider the angular variables $\theta$ and
$\alpha$. If we use the variable $\alpha$ as independent variable
the new differential system would not have the form for applying the
statements of Theorem \ref{secondave}. Hence,  since  $\epsilon>0$
is sufficiently small we have that $\dot{\theta}<0$, we can change
to the $\theta$ variable as the independent one. We denote by a
prime the derivative with respect to $\theta$. We obtain a new
system of three differential equations by dividing system
(\ref{ecrrho}) by $\dot{\theta}$
\begin{eqnarray}\label{ecrrho1}
 r'&=&\epsilon r^2\sin\theta\cos^2\theta\Big[ar^2\cos^2\theta+3b\rho^2\cos^2(\alpha+\theta)\Big]\\\nonumber
 &-& \frac{\epsilon^2}{4}r^3\sin\theta\cos^5\theta\Big[ar^2(1+\cos2\theta)+3b\rho^2(1+\cos2(\alpha+\theta))\Big]^2   +O(\epsilon^3),\\\nonumber
\rho'&=& \epsilon br^3
\rho\cos^3\theta\sin2(\alpha+\theta)-\epsilon^2br^4\rho\cos^6\theta\sin2(\alpha+\theta)\Big[ar^2\cos^2\theta\\\nonumber
&+&3b\rho^2\cos^2(\alpha+\theta)\Big]+O(\epsilon^3),\\\nonumber
\alpha'&=& -\epsilon
r\cos^3\theta\Big[ar^2\cos^2\theta+b(3\rho^2-2r^2)\cos^2(\alpha+\theta)\Big]\\\nonumber
&+&\epsilon^2r\cos^3\theta\Big[ar^2\cos^2\theta+b(3\rho^2-2r^2)\cos^2(\alpha+\theta)\Big]\Big[
ar^3\cos^5\theta\\\nonumber &+& 3b
r\rho^2\cos^3\theta\cos^2(\alpha+\theta)\Big] +O(\epsilon^3).
\end{eqnarray}

Now, the system (\ref{ecrrho1}) is 2$\pi$-periodic in the variable
$\theta$. In order to apply Theorem \ref{secondave} we fix the value
of the first integral at $H(r,\theta,\rho,\alpha)=h>0$.  By solving
equation (\ref{Hamrpo}) for $\rho$ we obtain
\begin{equation}\label{rho}
\rho=\sqrt{\frac{10h-5r^2-2\epsilon ar^5\cos^5\theta}{5(1+2\epsilon
br^3\cos^3\theta\cos^2(\alpha+\theta))}},
\end{equation}
notice that $\rho\to \sqrt{2h-r^2}$ when $\epsilon\to 0$. Expanding
(\ref{rho}) in Taylor series we obtain
\begin{eqnarray}\label{rhoexp}
\rho&=&\sqrt{2 h - r^2} - \frac{\epsilon r^3 \cos^3\theta}{10\sqrt{2
h -
r^2}}\Big[-10bh-ar^2(1+\cos2\theta)\\\nonumber&+&5br^2+5b(r^2-2h)\cos2(\alpha+\theta))
\Big] +O(\epsilon^2).
\end{eqnarray}
As we will apply averaging theory to first order, we can substitute
the zero order approximation of $\rho$ in equations (\ref{ecrrho1}),
which becomes
\begin{eqnarray}\label{ecrrho11}
r'&=&\epsilon r \sin\theta\cos^3\theta\Big[
ar^2\cos^2\theta+3b(2h-r^2)\cos^2(\alpha+\theta)\Big]\\\nonumber
&+&\frac{1}{20}\epsilon^2r^3\sin\theta\cos^5\theta\Big[12br^2\cos^2(\alpha+\theta)(-10bh\\\nonumber
&-&ar^2(1+\cos2\theta)+5br^2+5b(r^2-2h)\cos2(\alpha+\theta))\\\nonumber
&-&5(6bh+ar^2(1+\cos2\theta)-3br^2+b(6h-r^2)\cos2(\alpha+\theta))^2\Big]
+O(\epsilon^3),\\\nonumber \alpha'&=& \epsilon
r\cos^3\theta\Big[-ar^2\cos^2\theta+b(5r^2-6h)\cos^2(\alpha+\theta)\Big]\\\nonumber
&+&\epsilon^2\Big[r\cos^3\theta(ar^2\cos^2\theta+b(6h-5r^2)\cos^2(\alpha+\theta))\\\nonumber
&\cdot&(ar^3\cos^5\theta-3br(r^2-2h)\cos^3\theta\cos^2(\alpha+\theta))\\\nonumber
&-&\frac{3}{5}br^4\cos^6\theta\cos^2(\alpha+\theta)(-10bh-ar^2(1+\cos2\theta)+5br^2\\\nonumber
&+& 5b(r^2-2h)\cos2(\alpha+\theta))\Big]+O(\epsilon^3).
\end{eqnarray}
The system (\ref{ecrrho11}) has the canonical form  for applying the
averaging theory of second order and satisfies the assumptions for
$|\epsilon|
> 0$ sufficiently small, with $T = 2\pi$ and $F_1 = (F_{11},
F_{12})$ analytical functions. Averaging the function $F_1$ with
respect to the variable $\theta$ we obtain
$$f_1(r,\alpha)=\big(f_{11}(r,\alpha),f_{12}(r,\alpha)\big)=\int_{0}^{2\pi}(F_{11},F_{12})d\theta=(0,0),$$
hence, the averaging theory of first order does not apply because
the average functions of $F_1$ and $F_2$ vanish in the period.  We
proceed to calculate the function $f_2$ by applying the second order
averaging theory. The function $f_2$ is defined by
\begin{equation}
f_2(r,\alpha)=\int_{0}^{2\pi}\Big[D_{r\alpha}F_1(\theta,r,\alpha)\cdot
y_1(\theta,r,\alpha)+F_2(\theta,r,\alpha)\Big]dt,
\end{equation}
where
$$y_1(\theta,r,\alpha)=\int_{0}^\theta F_1(t,r,\alpha)=\left(\int_{0}^{\theta}F_{11}(t,r,\alpha)dt,\int_{0}^{\theta}F_{12}(t,r,\alpha)dt\right).$$
The two components of the vector function $y_1$ are
\begin{eqnarray}
y_{11}&=&\frac{1}{80}r^2\Big[-16ar^2(-1+\cos^5\theta)\\\nonumber&+&
b(2h-r^2)(40+8\cos2\alpha-30\cos\theta-10\cos3\theta\\\nonumber
&-&5\cos(2\alpha+3\theta)-3\cos(2\alpha+5\theta)-30\sin\theta\sin2\alpha)\Big],\\\nonumber
y_{12}&=&\frac{1}{240}r\Big[-ar^2(150\sin\theta+25\sin3\theta+3\sin5\theta)\\\nonumber
&-&
b(5r^2-6h)(6(5\cos\theta-8)\sin2\alpha)+30\sin\theta(2\cos2\alpha+3)\\\nonumber
&+&10\sin3\theta+15\sin(2\alpha+3\theta)+3\sin(2\alpha+5\theta))\Big].
\end{eqnarray}
From Theorem \ref{secondave} we arrive to the fact that the function
$f_2 = ( f_{21}, f_{22})$ is given by
\begin{eqnarray}\label{f12f22}
% \nonumber to remove numbering (before each equation)
  f_{21} &=& \frac{1}{320}br^3(r^2-2h)\sin2\alpha\big(-240bh+(-49a+50b)r^2\\\nonumber
  &+&90b(r^2-2h)\cos2\alpha\big),
  \\\nonumber
  f_{22} &=&\frac{1}{160}r^2(564b^2h^2+420abhr^2-678b^2hr^2+63a^2r^4-280abr^4\\\nonumber
  &+&170b^2r^4+b(49ar^2(3h-2r^2)+10b(48h^2-51hr^2+10r^4))\cos2\alpha\\\nonumber
  &+&45b^2(2h^2-3hr^2+r^4)\cos4\alpha).
\end{eqnarray}
The components of Jacobian determinant $J(r,\alpha)=\det
\left(D_{r,\alpha}f_2(r,\alpha)\right) $  are
\begin{eqnarray}\label{entries}
\frac{\partial f_{21}}{\partial r}&=&\Bigg[r^6\left(\frac{7b}{320}
(-49 a + 50 b) + \frac{63}{23} b^2 \cos2\alpha\right)\\\nonumber &-&
r^4 \left(\frac{15}{4}b^2h\left(1+\frac{3}{2}\cos2\alpha\right) +
\frac{bh}{32} (-49 a + 50 b) \right)\\\nonumber &+&
 \frac{9 b^2 h^2r^2}{2} \left( 1+ \frac{3}{4}  \cos2\alpha\right)\Bigg]\sin 2\alpha,\\\nonumber
\frac{\partial f_{21}}{\partial \alpha} &=& r^7
\Bigg(\frac{1}{160}(-49 a + 50 b) \cos2\alpha +
   \frac{9}{16} b^2 \cos^22\alpha\\\nonumber &-& \frac{9}{16} b^2 \sin^22\alpha\Bigg)
   +r^5 \Bigg(-\frac{3}{2} b^2 h \cos2\alpha\\\nonumber
   &-&  \frac{1}{80} b (-49 a + 50 b) h \cos2\alpha-
   \frac{9}{4} b^2 h \cos4\alpha \Bigg),\\\nonumber
\frac{\partial f_{22}}{\partial r} &=&r^5 \Bigg(\frac{189
a^2}{80}-\frac{21 a b}{2} + \frac{51 b^2}{8} -\frac{147}{40}
    a b \cos2\alpha\\\nonumber &+&\frac{15}{4}  b^2 \cos2\alpha +
   \frac{27}{16} b^2 \cos4\alpha\Bigg)\\\nonumber
   &+&r^3 \Bigg(\frac{21 a b h}{2} -\frac{339 b^2 h}{20}  + \frac{147}{40} a b h \cos2\alpha\\\nonumber &-&
   \frac{51}{4} b^2 h \cos2\alpha - \frac{27}{8} b^2 h \cos4\alpha\Bigg)\\\nonumber
   &+&r \left(\frac{141 b^2 h^2}{20} + 6 b^2 h^2 \cos2\alpha +
   \frac{9}{8}b^2 h^2 \cos4\alpha\right),\\\nonumber
\frac{\partial f_{22}}{\partial r} &=&\frac{1}{160}\Bigg( r^6
\left((196 a b-200 b^2) \sin2\alpha -
   180 b^2 \sin4\alpha\right)\\\nonumber
   &+& r^4 ((-294 a b h  + 1020 b^2 h) \sin2\alpha +
   540 b^2 h \sin4\alpha)\\\nonumber
   &+& r^2 (-960 b^2 h^2 \sin2\alpha - 360 b^2 h^2 \sin4\alpha)\Bigg).
\end{eqnarray}

\section{Existence of periodic solutions}\label{secperiodic}

We seek the values of $r^*$ and $\alpha^*$ such that
$f_{21}(r^*,\alpha^*)=0$, $f_{22}(r^*,\alpha^*)=0$ and
$J(r,^*\alpha^*)\neq0$; the value of $\rho^*=\sqrt{2h-(r^*)^2}$ is
obtained of (\ref{rhoexp}) for $\epsilon=0$. The values
$(r^*,\alpha^*,\rho^*)$ and the energy relation  (\ref{Hamrpo})
provide the initial conditions of the periodic solutions of the
system (\ref{ecrrho}) for $h>0$ sufficiently small in coordinates
$(r,\alpha,\rho,\theta)$. It is important to remark that the
periodic solutions exist for a small positive level of energy, since
the dynamics of our model is slightly different from the dynamics of
the harmonic oscillator, when the energy increases we obtain more
complex motions (Marchesiello  and Puccaco 2011).
\bigskip

Starting with equation (\ref{f12f22}), if $r^*=0,\sqrt{2h}$, or
$\alpha^*=0, \pm\frac{\pi}{2}, \pi$ then $f_{21}(r,\alpha)=0$. Using
the above values of $r^*$ and $\alpha^*$ we seek the zeros of
equation $f_{22}(r,\alpha)= 0$ in every case; then we evaluate the
Jacobian determinant $J(r,^*\alpha^*)$. For $r^*=0$ it easy to
verify that $J(0,\alpha)=0$ for all $\alpha$, this  case is not a
good solution. After calculations, we have the following proposition
for the solutions $(r^*, \alpha^*, \rho^*)$ of the system
(\ref{f12f22}).

\begin{proposition}\label{sol1}
For $h>0$ small,  and $a$, $b$ real numbers we have that
\begin{enumerate}
\item If $\Big|\frac{9a^2-10ab-4b^2}{b(7a+10b)}\Big|< \frac{1}{2}$ and $(a - 2 b) (9 a - 2 b) (2 a - b) (2 a + b) h^6\neq0$
 there are two  solutions of system (\ref{f12f22}) given by
 $$\left(\sqrt{2h},
\pm\frac{1}{2}\arccos\left(\frac{2(9a^2-10ab-4b^2)}{b(7a+10b)}\right),0
\right).$$
\item If $(2a+b)(-a+b)<0$, $3b(-a+b)>0$ and $a b^6 (2 a + b) (a +3 b) h^6\neq0$  there are
two  solutions of the system (\ref{f12f22}) given by
$$\left(\sqrt{\frac{3bh}{-a+b}}, 0
,\sqrt{-\frac{h(2a+b)}{-a+b}}\right),\hspace{3mm}\left(\sqrt{\frac{3bh}{-a+b}},
\pi ,\sqrt{-\frac{h(2a+b)}{-a+b}} \right).$$
\item If $(-a+2b)(-a+5b)<0$, $b(-a+5b)>0$ and $(a - 10 b) (a - 2 b)(a + 3 b)b^6 h^6\neq0$
there are two  solutions of (\ref{f12f22}) given by
$$\left(\sqrt{\frac{6bh}{-a+5b}}, 0,
\sqrt{\frac{2h(-a+2b)}{-a+5b}}\right),\hspace{3mm}\left(\sqrt{\frac{6bh}{-a+5b}},
\pi, \sqrt{\frac{2h(-a+2b)}{-a+5b}}\right).$$
\item If $63 a^2-182ab+115b^2<0$ and $273 a b h - 303 b^2 h>0$
 or  $63 a^2-182ab+115b^2<0$ and $273 a b h - 303 b^2 h<0$
there is one solution of (\ref{f12f22}).
\item If $63 a^2-182ab+115b^2>0$ and $273 a b h - 303 b^2 h<0$
there are two solutions of  (\ref{f12f22}).
\end{enumerate}
\end{proposition}
\begin{proof}
The solutions (1)-(3) of the system are calculated directly
substituting the values of  $r^*=0$, $r^*=\sqrt{2h}$ and
$\alpha^*=0,$ $\alpha^*=\pi$, respectively.  For
$\alpha^*=\pm\frac{\pi}{2}$ we obtain the equation
\begin{equation}\label{r4}
f_{22}\left(r,\pm\frac{\pi}{2}\right)=174 b^2 h^2  + (273 a b h -
303 b^2 h) r^2 + (63 a^2 - 182 a b +
    115 b^2) r^4=0.
\end{equation}
In order to seek  the positive roots of (\ref{r4}) we introduce
$u=r^2$, the
 equation (\ref{r4}) is rewritten as
\begin{equation}\label{ur2}
174 b^2 h^2  + (273 a b h - 303 b^2 h) u + (63 a^2 - 182 a b +
    115 b^2) u^2=0.
\end{equation}
The discriminant of equation (\ref{ur2}) is $\Delta=3 b^2h^2 (10227
a^2 - 12922 a b + 3923 b^2) $, in order to have real solutions
$\Delta\ge0$. Using the Descartes' rule of signs  the number of
positive roots $u^*$ of equation (\ref{r4})  is given in the
following list
\begin{itemize}
\item If $63 a^2-182ab+115b^2>0$ and $273 a b h - 303 b^2 h>0$
there are no positive roots of equation (\ref{ur2}).
\item If $63 a^2-182ab+115b^2<0$ and $273 a b h - 303 b^2 h>0$
there is one positive root of equation (\ref{ur2}).
\item If $63 a^2-182ab+115b^2<0$ and $273 a b h - 303 b^2 h<0$
there is one positive root of equation (\ref{ur2}) .
\item If $63 a^2-182ab+115b^2>0$ and $273 a b h - 303 b^2 h<0$
there are two positive roots of the (\ref{ur2}).
\end{itemize}
Since we are interested in positive roots $r^*$ of equation
(\ref{r4}) we consider $r^*=\sqrt{u^*}$, hence, for each positive
value of $u^*$ there is a positive value of $r^*$. Now, we proceed
to proof that the positive roots of (\ref{r4}) satisfy that
$J(r^*,\pm\frac{\pi}{2})\neq0$. The Jacobian determinant evaluated
at $\alpha^*=\pm\frac{\pi}{2}$ is
\begin{equation}\label{jac0}
J\left(r,\pm\frac{\pi}{2}\right)=\frac{1323}{12800}(r^2-2 h) (60 b h
+ (7 a  - 20 b )r^2) (6 b^2 h^2 +
   2 (3 a - 7 b) b h r^2 + (a - 5 b) (a - b) r^4),
\end{equation}
clearly,  $r=\sqrt{2h}$, $\sqrt{\frac{60bh}{-7a + 20 b}}$ are roots
of $J\left(r,\pm\frac{\pi}{2}\right)=0$, but these values do not
provide solutions of (\ref{f12f22}). On the other hand,  the
remaining roots of (\ref{jac0}) are obtained from the equation
\begin{equation}\label{jac}
6 b^2 h^2 + 2 (3 a - 7 b) b h r^2 + (a - 5 b) (a - b) r^4=0.
\end{equation}
It is easy to see that the positive roots of equation (\ref{r4}) are
not roots of equation (\ref{jac}), hence, if $r^*$ is a positive
root of (\ref{r4}) then  $J(r^*,\pm\frac{\pi}{2})\neq0$.
\end{proof}

In Proposition \ref{sol1} the averaging method provided the initial
conditions for the periodic orbits of system (\ref{HSys1}) in
coordinates $(r,\alpha,\rho,\theta)$ which are the solutions $(r^*,
\alpha^*, \rho^*)$ of the system (\ref{f12f22}) in each case. The
corresponding value of $\theta^*$ is calculated directly from the
energy relation (\ref{Hamrpo}) for small values of $h>0$ using the
above values of $(r^*, \alpha^*, \rho^*)$. In order to characterize
the periodic orbits in the original variables $(x,y,p_x,p_y)$ we use
the change of coordinates (\ref{coordchange}). The initial
conditions for each type of periodic solution are given in the
following Proposition.

\begin{proposition}\label{initialcond}
For $h>0$ positive small the initial conditions
$(x^*,y^*,{p_x}^*,{p_y}^*)$ of the periodic orbits of system
(\ref{HSys1})  satisfy the following statements
\begin{enumerate}
\item If $\Big|\frac{9a^2-10ab-4b^2}{b(7a+10b)}\Big|< \frac{1}{2}$ and $(a - 2 b) (9 a - 2 b) (2 a - b)
 (2 a + b) h^6\neq0$ we have the initial condition $(x^*,y^*,{p_x}^*,{p_y}^*)=(\sqrt{2h}\cos\theta^*,
 \sqrt{2h}\sin\theta^*,0,0)$ of one periodic orbit.
\item If $(2a+b)(-a+b)<0$, $3b(-a+b)>0$ and $a b^6 (2 a + b) (a +3 b)
h^6\neq0$ we have the initial condition $(x^*,y^*,{p_x}^*,{p_y}^*)$
given by
$$\left(\sqrt{\frac{3bh}{-a+b}}\cos\theta^*,\sqrt{\frac{3bh}{-a+b}}\sin\theta^*,
\sqrt{-\frac{h(2a+b)}{-a+b}}\cos\theta^*,\sqrt{-\frac{h(2a+b)}{-a+b}}\sin\theta^*\right),$$
of one periodic orbit.
\item If $(-a+2b)(-a+5b)<0$, $b(-a+5b)>0$ and $(a - 10 b) (a - 2 b)(a + 3 b)b^6 h^6\neq0$ we have the initial
condition $(x^*,y^*,{p_x}^*,{p_y}^*)$ given by
 $$\left(\sqrt{\frac{6bh}{-a+5b}}\cos\theta^*,\sqrt{\frac{6bh}{-a+5b}}\sin\theta^*,
\sqrt{\frac{2h(-a+2b)}{-a+5b}}\cos\theta^*,\sqrt{\frac{2h(-a+2b)}{-a+5b}}\sin\theta^*\right),$$
of one periodic orbit.
\item If $r_0^*$ is a positive solution of equation (\ref{r4}) there are two initial conditions of two periodic solutions
 $$(x^*,y^*,{p_x}^*,{p_y}^*)=(r_0^*\cos
\theta^*,r_0^*\sin\theta^*,
-\sqrt{2h-(r_0^*)^2}\sin\theta^*,\sqrt{2h-(r_0^*)^2}\cos\theta^*
),$$
$$(x^*,y^*,{p_x}^*,{p_y}^*)=(r_0^*\cos \theta^*,r_0^*\sin\theta^*,
\sqrt{2h-(r_0^*)^2}\sin\theta^*,-\sqrt{2h-(r_0^*)^2}\cos\theta^*
).$$
\end{enumerate}
 The value of $\theta^*$ is obtained from energy relation (\ref{HSys1}) for the initial
conditions $(r^*,\alpha^*,\rho^*)$ and the value of the energy $h$.
\end{proposition}
\begin{proof}
The initial conditions are obtained directly from the coordinate
change (\ref{coordchange})
$$x=r\cos\theta,
\hspace{3mm}p_x=r\sin\theta,\hspace{3mm}y=\rho\cos(\theta+\alpha),\hspace{3mm}p_y=\rho\sin(\theta+\alpha),$$
we proof that the two solutions $(r^*,\alpha^*, \rho^*)$ obtained in
(1), (2) and (3) in Proposition \ref{sol1}  are
 initial conditions of one periodic orbit in each case; also we
proof that the conditions in (4) in Proposition \ref{sol1} give two
initial conditions for two periodic solutions.
\begin{enumerate}
\item For $r^*=\sqrt{2h}$, $\alpha=\pm\frac{1}{2}\arccos\left(\frac{2(9a^2-10ab-4b^2)}{b(7a+10b)}\right)$
we have that $\rho^*=0$, then we obtained
one initial condition and one periodic orbit.
\item For $r^*=\sqrt{\frac{3bh}{-a+b}}$ and $\rho^*=\sqrt{-\frac{h(2a+b)}{-a+b}}$, if $\alpha=0$ and $\alpha=\pi$ we obtained two initial
conditions
$$\left(\sqrt{\frac{3bh}{-a+b}}\cos\theta^*,\sqrt{\frac{3bh}{-a+b}}\sin\theta^*,
\pm\sqrt{-\frac{h(2a+b)}{-a+b}}\cos\theta^*,\pm\sqrt{-\frac{h(2a+b)}{-a+b}}\sin\theta^*\right),$$
since the symmetry conditions which satisfies the variables $y$ and
$p_y$ the above initial conditions determine the same orbit.
\item  For $r^*=\sqrt{\frac{6bh}{-a+5b}}$ and $\rho^*=\sqrt{\frac{2h(-a+2b)}{-a+5b}}$, if $\alpha=0$ and $\alpha=\pi$ we obtained two initial
conditions
$$\left(\sqrt{\frac{6bh}{-a+5b}}\cos\theta^*,\sqrt{\frac{6bh}{-a+5b}}\sin\theta^*,
\pm\sqrt{\frac{2h(-a+2b)}{-a+5b}}\cos\theta^*,\pm\sqrt{\frac{2h(-a+2b)}{-a+5b}}\sin\theta^*\right).$$
since the symmetry conditions which satisfies the variables $y$ and
$p_y$ the above initial conditions determine the same orbit.
\item For $\alpha=\frac{\pi}{2}$ and $\alpha=-\frac{\pi}{2}$ the symmetry conditions of variables $y$ and $p_y$ are not satisfied,
 then we have two orbits for each $r_0^*$ positive solution of equation (\ref{r4}).
\end{enumerate}
\end{proof}
According to Section \ref{eqmotion}, after a scaling with the
parameter $\epsilon$ we can obtain the solutions in coordinates
$(X,Y,P_X,P_Y)$ .
\bigskip

The main result on the periodic orbits of the system
(\ref{Hamiltonian}) is summarized as follows.
\begin{theorem}\label{teoprin}
 For the Hamiltonian system
(\ref{HSys1}) with small  energy  $h>0$, we have the following
statements
\begin{enumerate}
\item If $\Big|\frac{9a^2-10ab-4b^2}{b(7a+20b)}\Big|< \frac{1}{2}$ and $(a - 2 b) (9 a - 2 b) (2 a - b) (2 a + b) h^6\neq0$
 there is one unstable  periodic solution.
\item If $(2a+b)(-a+b)<0$, $3b(-a+b)>0$ and $a b^6 (2 a + b) (a +3 b) h^6\neq0$  there
is one unstable periodic solution.
\item If $(-a+2b)(-a+5b)<0$, $b(-a+5b)>0$ and $(a - 10 b) (a - 2 b)(a + 3 b)b^6 h^6\neq0$
there is one unstable periodic solution.
\item If $63 a^2-182ab+115b^2<0$ and $273 a b h - 303 b^2 h<0$
there are two periodic solutions.
\item If $63 a^2-182ab+115b^2>0$ and $273 a b h - 303 b^2 h<0$
there are four periodic solutions.
\end{enumerate}
The initial conditions $(x^*,y^*,{p_x}^*,{p_y}^*)$ for each type of
periodic solutions are given in Proposition \ref{initialcond}.
\end{theorem}
\begin{proof}
From Proposition \ref{initialcond}, it follows number of periodic
solutions of the system (\ref{Hamiltonian}) in each case. The kind
of stability of the  periodic solutions is given by  the sign of the
eigenvalues $\lambda_1$ and $\lambda_2$   of the Jacobian matrix
$J(r^*,\alpha^*)$ whose entries are given by (\ref{entries}). It can
be verified  that $J(r^*,\alpha^*)=\lambda_1\lambda_2$ in every
case. According with the hyphotesis of this theorem we have that:
\bigskip

For (1) the eigenvalues are real and satisfy that
$\lambda_1=-\lambda_2$, the family of periodic solutions is
unstable.
\bigskip

For (2) the eigenvalues are
$$\lambda_1=\frac{819}{160}\sqrt{3} a (2
a + b)\left(\frac{-b h}{a - b}\right)^{\frac{7}{2}}, \hspace{2mm}
\lambda_2=\frac{567}{80}\sqrt{3} b (a + 3 b) h \left(\frac{b h}{-a +
 b}\right)^{\frac{3}{2}},$$
by hypothesis the parameters satisfy $b(2a+b)<0$, there are no
values of $a$ and $b$ such that  $\lambda_1<0$ and $\lambda_2<0$,
then the family of periodic solutions is unstable.
\bigskip

For (3) the eigenvalues are
$$\lambda_1=\frac{189}{10} \sqrt{\frac{3}{2}} (a - 10 b) (a - 2 b) \left(\frac{-b h}{a - 5 b}\right)^{\frac{7}{2}},\hspace{2mm}
\lambda_2=-\frac{567}{20} \sqrt{\frac{3}{2}} b (a + 3 b) h
\left(\frac{-b h}{a - 5 b}\right)^{\frac{3}{2}},$$ by hypothesis the
arameters satisfy $b(2a+b)<0$, there are no values of $a$ and $b$
such that  $\lambda_1$ and $\lambda_2<0$, then the family of
periodic solutions is  unstable.
\bigskip

From $J\left(r,\pm\frac{\pi}{2}\right)$ in (\ref{jac0}) the
eigenvalues for (4) and (5) are given by
\begin{equation*}
\lambda=\pm\sqrt{-J\left(r_0^*,\pm\frac{\pi}{2}\right)},
\end{equation*}
where $r_0^*$ is given by Proposition \ref{initialcond}, hence,
depending of the values of $a$, $b$ and $h$ the orbits are stable or
unstable.
\end{proof}

\section{Conclusion}
The existence of periodic solutions for Hamiltonian systems with
polynomial homogeneous terms of fifth degree with two real
parameters $a$ and $b$ is established. The averaging method second
order can be apply to study the existence of periodic solutions for
a more general Hamiltonian  system with three or four parameters
\begin{equation}
 H=\frac{1}{2}\left(
    P_X^2+P_Y^2\right)+\frac{1}{2}\left(X^2+Y^2\right)+\frac{a}{5}X^5+bX^3Y^2+ c
    X^4Y.
\end{equation}
It is possible to obtain the number of periodic solutions in terms
of the parameters in some cases as in the Theorem \ref{teoprin}. The
main problem  arises, as shown in the proof of Proposition
\ref{sol1}, when for certain values of $\alpha^*$ the polynomial
equations $f_{22}(r,\alpha^*)=0$  obtained  have degree greater than
six, for this reason, it can not be possible to give explicit
conditions in terms three parameters.

\subsection*{Acknowledgment}
 The author
is pleased to acknowledge the financial support from DGAPA which
allows him a postdoctoral stay in the department of mathematics of
the faculty of sciences, UNAM. Published in Astrophysics and Space
Science DOI: 10.1007/s10509-015-2612-0


\begin{thebibliography}{99}
\bibitem{Puccaco}
Belmonte C.,Boccaletti D., Puccaco G. :Quantitative predictions with
detuned normal forms, Celest. Mech. Dyn. Astron. \textbf{102},
163-176 (2008)

\bibitem{Boccaletti}
Boccaletti D., Puccaco G.: Theory of Orbits, Vol. 2: Perturbative
and geometric methods, Springer-Verlag, Berlin (1999)

\bibitem{buica}
Buica A., Llibre J.: Averaging methods for finding periodic orbits
via Brouwer degree. Bull. Sci. Math. \textbf{128}, 7--22 (2004)

\bibitem{Caranicolas}
 Caranicolas N.,  Vozikis Ch.: Chaos in a quartic dynamical model.
 Cel. Mech. \textbf{40}, 35--49 (1987)

\bibitem{Caranicolas Vozikis}
 Caranicolas N.,  Vozikis Ch.: Order and chaos in galactic maps.
 Astron. Astrophys.  \textbf{349}, 70--76 (1999)

\bibitem{Vidal}
Carrasco D., Vidal C.: Periodic solutions, stability and
non-integrability in a generalized H\'enon-Heiles Hamiltonian
System. J. of Nonlinear Physics, \textbf{20}  199-21 (2013)

\bibitem{Contopulos}
 Contopulos G.: Orden and chaos in Dynamical Astronomy.
 Astron. Astrophys. Library. Springer: Berlin, Heidelberg, New York (2002)

\bibitem{Gram}
Dorizzi B.,  Grammaticos B.: Integrability of Hamiltonians with
third and fourth degree polynomial potentials. J. Math. Phys.
\textbf{24} (9),  2289--2295 (1983)

\bibitem {Falconi1}  Falconi M.,
Lacomba E. A.: On the dynamics of mechanical systems with the
homogeneous polynomial potential $V=ax^4+cx^2y^2$. JDDE \textbf{21},
527--554 (2009)

\bibitem{Falconi2}
 Falconi M., Lacomba E. A.: Asymptotic behavior of escape solutions of mechanical systems
 with homogeneous potentials. Cont. Math. \textbf{198},
 181--195 (1996)

\bibitem {Falconi3}
Falconi M.,  Lacomba E. A.,  Vidal C.: On the dynamics of mechanical
systems with homogeneous polynomial potentials of degree 4. Bull.
Braz. Math. Soc. New Series \textbf{38}(2), 301 (2007)



\bibitem{HH}
Heiles C.,  H\'enon  M.: The applicability of the third integral of
motion: some numerical experiments. Astron.J. \textbf{69}, 73-84
(1964)

\bibitem{Lidia}
Jim\'enez-Lara L.,  Llibre J.: Periodic orbits and nonintegrability
of generalized classical Yang-Mills Hamiltonian Systems, J.  Math.
Phys. \textbf{52},  205103-14 (2011)

\bibitem{Llibre}
Llibre J., Makhlouf A.: Periodic orbits of the generalized
Friedmann-Robertson-Walker Hamiltonian systems, Astrophysics Space
Sci \textbf{344}, 45-50 (2013)

\bibitem{Marchesiello}
 Marchesiello A., G. Pucacco G.: Relevance of the 1:1 resonance in galactic
 dynamics, The European Physical Journal Plus \textbf{126}, article
 id.104 (2011)


\end{thebibliography}
\end{document}